\documentclass[superscriptaddress,prx,nofootinbib,twocolumn]{revtex4-2}
\usepackage{amsmath,amssymb,amsthm,amsfonts}
\usepackage{mathbbol}
\usepackage{graphicx}
\usepackage[percent]{overpic}
\usepackage{times}
\usepackage{hyperref}
\hypersetup{
colorlinks=true,
linkcolor=red,
citecolor=red,}
\usepackage{textcomp}
\usepackage{enumerate}
\usepackage{multirow}
\usepackage{tabularx}
\usepackage[mathscr]{euscript}
\usepackage{mathtools}
\usepackage{mhchem}
\usepackage[mathscr]{euscript}
\usepackage[normalem]{ulem}
\usepackage{centernot}
\usepackage{stmaryrd}
\usepackage{longtable}
\usepackage{enumitem,kantlipsum}
\usepackage{array}
\usepackage{wrapfig}
\usepackage{tabu}
\usepackage{graphicx}
\usepackage{placeins}
\usepackage{enumitem}
\usepackage{dutchcal}
\usepackage{tabularx}
\usepackage[export]{adjustbox}
\usepackage{color,soul}
\usepackage{longtable}
\usepackage[dvipsnames]{xcolor}
\usepackage{float}
\usepackage[caption=false]{subfig}
\usepackage{adjustbox}

\usepackage{tikz}
\usetikzlibrary{quantikz}

\newtheorem*{theorem*}{Theorem}

\theoremstyle{definition}

\newcommand\s[1]{_{\rm #1}}

\newcommand{\ketbra}[1]{ | #1 \rangle\!\langle #1 |}

\newcommand{\rank} {\operatorname{rank}}

\newcommand{\inprod}[2] {\langle #1 , #2 \rangle}

\newcommand{\one}{\leavevmode\hbox{\small1\normalsize\kern-.33em1}}

\newcommand{\R}{\mathbb{R}}

\newcommand{\M}{{\cal M}}

\newcommand{\PP}{{\cal P}}

\newcommand{\EE}{{\cal E}}

\newcommand{\St}{{\cal S}}
\newcommand{\V}{{\cal V}}

\newcommand{\MM}{\mathsf{M}}
\newcommand{\pp}{\mathsf{P}}
\newcommand{\e}{\mathsf{E}}

\def\R{{\mathbb{R}}}

\def\spn{\mathord{\rm span}}

\begin{document}

\title{Contextuality of all optimal quantum cloning}
\author{Mina Doosti}
\affiliation{School of Informatics, Quantum Software Lab, University of Edinburgh, United Kingdom}
\author{Theodoros Yianni}
\affiliation{Department of Computer Science, Royal Holloway, University of London, United Kingdom}
\author{Farid Shahandeh}
\affiliation{Department of Computer Science, Royal Holloway, University of London, United Kingdom}


\begin{abstract}
Quantum contextuality is a key nonclassical feature underlying advantages in quantum computation and communication. We introduce a new method to study contextuality in quantum information–processing tasks and protocols, relying solely on observed information-processing statistics. Building on the framework of~\cite{shahandeh2025unified} and employing rank separation techniques, we prove that contextuality is the necessary resource in both phase-covariant and universal optimal quantum cloning, thereby establishing its role as a fundamental source of nonclassicality in all known optimal cloning scenarios and resolving an open problem on the connection between cloning and contextuality. As a second application, we demonstrate the power of our method by providing a new, streamlined proof of contextuality in minimum-error quantum state discrimination.
\end{abstract}

\maketitle

Understanding what truly makes quantum theory nonclassical lies at the heart of quantum information science. From foundational principles to practical protocols and algorithms in computation, communication, and cryptography, the field has been driven by the quest to identify the physical resources that distinguish quantum from classical information processing. A paradigmatic manifestation of such nonclassicality is the no-cloning theorem, which states that, unlike classical bits, unknown quantum states cannot be copied exactly and deterministically~\cite{Wootters1982}.
This fundamental difference is often viewed as a unique nonclassical feature of quantum mechanics.
Its impact also extends to various security proofs and reductions in the quantum realm~\cite{scarani2005quantum,broadbent2016quantum,aaronson2009quantum}. 

Another major notion of nonclassicality is quantum contextuality.
In its various formulations~\cite{Kochen1967,spekkens2005,Cabello2008,Abramsky2011}, contextuality generally refers to the impossibility of reproducing all quantum statistics with a classical hidden-variable model that assigns outcomes independently of the measurement context. Contextuality is therefore regarded as an underlying resource for many quantum advantages in computation and communication~\cite{Raussendorf2013,Bermejo-Vega2017,Frembs2018,Shahandeh2021arXiv,Schmid2022,Montina2014,Ambainis2016,Gupta2023}. These considerations naturally lead to the question of how these two fundamental notions are linked. 

The no-cloning theorem has also been shown to hold beyond quantum theory in general probabilistic theories~\cite{barnum2006cloning}. It is crucial, however, to recognize that the exact no-cloning theorem only applies to overlapping quantum states with their classical counterparts being overlapping probability distributions, which is also impossible to mimic perfectly~\cite{hardy1999,Spekkens_toy}.
Furthermore, quantum cloning machines can produce approximate copies with bounded fidelities~\cite{Buzek96,Bruss1998}.
The latter is a more restricted notion of no-cloning, prompting the question of whether one can link it to a fundamental nonclassical feature like contextuality. 

It was shown in Ref.~\cite{lostaglio2020contextual} that the optimal state-dependent cloning fidelity predicted by quantum theory does not admit a noncontextual ontological explanation.
However, the problem remained open for more general schemes, namely, phase-covariant and universal cloning.
The main obstacle here is that the input states do not parameterize the optimal cloners~\cite{Bruss1998}.

One of the main limitations of previous approaches is that, with the existing methods and proof techniques, it has been difficult to pinpoint the role of contextuality in quantum information–processing tasks directly from their observable data, which in turn obstructed the establishment of general results. However, recently in Ref.~\cite{shahandeh2025unified}, the authors propose a novel approach to generalized contextuality~\cite{spekkens2005} and introduce the \textit{rank-separation} criterion.
This criterion is solely based on statistics and does not require analyzing operational equivalences and identities, hence departing from the traditional methods. This new characterization of contextuality can serve as a powerful tool for the study of quantum resources required in operational tasks and quantum protocols, which is otherwise challenging to study in the other framework, apart from toy models and contrived examples.

In this Letter, we initiate the study of the rank-separation criterion for characterizing contextuality in fundamental yet practical quantum tasks, and prove a unifying theorem establishing the contextuality of all quantum cloning machines, including the previously open phase-covariant and universal cases~\cite{Bruss1998}. Our result shows that optimal cloning is necessarily contextual, thereby identifying contextuality as a fundamental resource underlying all optimal quantum cloning schemes and resolving this open problem. As a second application, highlighting the generality of our approach, we provide a new proof of contextuality for minimum-error quantum state discrimination. Although contextuality in these tasks has been studied by other techniques~\cite{Schmid2018,Bruss1998,lostaglio2020contextual}, our argument is substantially simpler and offers clearer conceptual insight.

The cloning task is defined as processing a state $\rho$, chosen from a set $\St$ of inputs, and a fixed ancilla (blank) state by a physical map producing a bipartite output intended to approximate two copies of $\rho$. The choice of $\St$ depends on the cloning task: in universal cloning, it contains all pure states on the Bloch sphere (or the full Hilbert space in higher dimensions), in phase-covariant cloning, it is restricted to equatorial states, while in state-dependent cloning it consists of two fixed nonorthogonal states. The performance of cloning is quantified by the \emph{cloning fidelity}, defined as the overlap between the actual bipartite output state $\sigma$ (or its reduced subsystems) and the ideal cloned state $\rho\otimes\rho$. A map achieving the maximum fidelity allowed by quantum theory is called an \emph{optimal} quantum cloner.

While the literature typically focuses on characterizing the cloning map, performance depends only on the output states relative to the ideal clones. Cloning can therefore be viewed as a prepare-measure scenario, with measurements chosen so that the resulting statistics fully determine the cloning fidelity. Concretely, this is done by measuring in the bases of the ideal and actual clones, so that the success and error probabilities directly fix the fidelity.

In an operational description of a prepare-measure scenario, there are $I$ preparations and $J$ measurements, each of which is a laboratory instruction for preparing and measuring the system, respectively.
We denote the collections of these instructions by $\PP:=\{\pp_i\}$ and $\M:=\{\MM_j\}$.
We also assume that each measurement results in $n$ outcomes, each denoted by $\e_k$.
The observable probabilities in a prepare-measure scenario can then be recast in a $I\times Jn$ matrix $C\s{F}$,
\begin{gather}\label{eq:COPE}
C\s{F} := 
\begin{pmatrix}
\begin{array}{c}
   \begin{matrix}
p(1|\pp_1,\MM_1) & \cdots & p(1|\pp_I,\MM_1) \\
        \vdots &   & \vdots  \\
        p(n|\pp_1,\MM_1) & \cdots & p(n|\pp_I,\MM_1)
   \end{matrix}\\ 
   \hline
    \vdots \\
   \hline
   \begin{matrix}
        p(1|\pp_1,\MM_J) & \cdots & p(1|\pp_I,\MM_J)\\
        \vdots & & \vdots \\
        p(n|\pp_1,\MM_J) & \cdots & p(n|\pp_I,\MM_J)
   \end{matrix} \\
\end{array}    
\end{pmatrix}.
\end{gather}
Tsirelson referred to $C\s{F}$ as the \textit{behavior} of the system~\cite{Tsirelson1993}. 
In our framework, however, this matrix is viewed as the restriction of the complete probabilistic structure of an operational theory, i.e., its matrix of \textit{conditional outcome probabilities of events} (COPE), to subsets of preparations and measurements, and is referred to as the \textit{fragment} COPE matrix~\cite{shahandeh2025unified}.

Contextuality is usually analyzed at the level of the operational theory rather than individual tasks.
The challenge, then, is to establish when the theory’s contextual features actually serve as a resource for a given task.
To overcome this hurdle, let us recap the noncontextuality of operational theories as given in Ref.~\cite{shahandeh2025unified}.

First, all COPE matrices can be modeled using \textit{ontological models}, in which an \textit{ontic} variable space $\Lambda$ is assumed to underlie physical phenomena.
It is assumed that there is a countable and finite number of ontic points, hence $\Lambda$ is embedded in a vector space $\ V\s{ont}\!\cong\!\R^s$ with $s\!:=\!|\Lambda|$.
The preparations and measurement outcomes are assigned to probability distributions and response functions over $\Lambda$, respectively.
The probability distributions are called the \textit{epistemic states} (ES) and denoted by $\mu_i=(\mu_{\lambda i})_\lambda$ such that $0\!\leq\!\mu_{\lambda i}\!\leq 1$ and $\sum_\lambda\! \mu_{\lambda i}\!=\!1$ where the sum is over all ontic points.
The response functions (RF) for a measurement $\MM_j$ are denoted by $\xi^j_k\!=\!(\xi^j_{k\lambda})_\lambda$ and satisfy $\xi^j_{k\lambda}\!\in\![0,1]$ and $\sum_k\!\xi^j_{k\lambda}\!=\!1$ at any point $\lambda$.
Then, the probability of a particular outcome $k$ given the preparation $i$ is obtained as $p(k|i,j)\!=\!\sum_\lambda \xi^j_{k\lambda}\mu_{\lambda i}\! =\! \inprod{\xi^j_k}{\mu_i}$, where the inner product is that of $\R^s$.
From this, it is not difficult to see that the problem of finding ontological models for COPE matrices of operational theories is equivalent to the \textit{nonnegative matrix factorization} (NMF) problem~\cite{shahandeh2025unified}. In other words, every ontological model corresponds to a decomposition of the COPE matrix, $C\!=\!RE$, where $R$ is an $N\!\times\!s$ nonnegative RF matrix and $E$ is an $s\!\times\!M$ nonnegative ES matrix.

Second, an ontological model is called noncontextual
if operationally equivalent procedures are represented identically in the model.
That is,
\begin{equation}\label{eq:ctxt_def}
    \pp_1\cong\pp_2 \Leftrightarrow \mu_{\pp_1} = \mu_{\pp_2},\quad
    \e_1 \cong \e_2 \Leftrightarrow \xi_{\e_1} = \xi_{\e_2},    
\end{equation}
where operational equivalence $\cong$ means that no viable measurement (preparation) can separate the preparations $\pp_1$ and $\pp_2$ (outcome events $\e_1$ and $\e_2$)~\cite{spekkens2005}.
This means that two operationally equivalent preparations correspond to an identical column of the COPE matrix.
Similarly, two operationally equivalent measurement outcomes, $\e_1$ and $\e_2$, correspond to identical rows in the COPE matrix.
Note that, in representing COPE matrices, we adopt an extremal-quotiented convention: only extremal rows and columns are shown, with repetitions of rows within a measurement and repetitions of columns omitted.~\cite{shahandeh2025unified}.

Thanks to the COPE formalism, the uniqueness of representations of Eq.~\eqref{eq:ctxt_def}, has a precise linear-algebraic implication: An ontological model is preparation (measurement) noncontextual if and only if its set of response functions (epistemic states) separates the points in the linear span of its epistemic states (response functions). 
It follows that~\cite{shahandeh2025unified} an operational theory with a COPE matrix $C$ admits a noncontextual ontological model if and only if it admits an NMF $C\!=\!RE$ such that
\begin{equation}\label{eq:nctxt_rank}
   \rank{C}=\rank{R}=\rank{E}.
\end{equation}
An NMF satisfying Eq.~\eqref{eq:nctxt_rank} is called an \textit{equirank} NMF (ENMF)~\cite{shahandeh2025unified,yianni2025}.
Although all COPE matrices admit an NMF, and thus an ontological model, an ENMF and a \textit{noncontextual} ontological model may cease to exist~\cite{shahandeh2025unified}.
The failure of the equirank condition is termed \textit{rank separation}.

The definition of contextuality at the level of the operational theory does not, by itself, determine how it manifests in concrete information‑processing tasks. 
A theory may exhibit contextuality without it being operationally activated by a particular task.
To determine when it truly serves as a resource, we must therefore translate the theory-level criteria, such as rank separation, into criteria on the task's statistics, i.e., the fragment COPE matrix.

A key role in the derivation of Eq.~\eqref{eq:nctxt_rank} is played by the \textit{completeness} assumption, which states that the COPE matrix encodes the entire probabilistic structure of the operational theory.
This implies that there are no preparations or measurement outcomes that increase the number of separated measurement outcomes or preparations, respectively.
This property is called the \textit{relative tomographic completeness}~\cite{schmid2024shadows,shahandeh2025unified}.
This, in turn, guarantees that adding any number of admissible preparations and measurements within the operational theory cannot transform a COPE matrix that does not admit an ENMF into one that does.

A fragment COPE matrix need not satisfy the completeness assumption. However, some fragments do satisfy relative tomographic completeness. In such cases, the fragment COPE matrix plays the role of a restricted operational theory, so rank separation in the fragment implies rank separation, and hence contextuality, of the parent theory.

As an example relevant here, consider the qubit operational theory and the fragment consisting of states and effects on the equator of the Bloch sphere. This fragment satisfies relative tomographic completeness. Therefore, if its fragment COPE matrix admits no ENMF, then neither does the full qubit theory. This leads to our first main result.

\paragraph*{Theorem~1.}
Suppose $C\s{F}$ is a fragment of a COPE matrix $C$ and satisfies relative tomographic completeness.
Then, if $C\s{F}$ does not admit an ENMF, neither does $C$.
\paragraph*{Proof.}
Assume the contrary, i.e., that $C=RE$ with nonnegative matrices $R$ and $E$ such that $\rank{C}\!=\!\rank{R}\!=\!\rank{E}$.
Since $C\s{F}$ is just a restriction of $C$, there are restrictions $R\s{F}$ and $E\s{F}$ of $R$ and $E$, respectively, such that $C\s{F}\!=\!R\s{F}E\s{F}$.
Now if $\rank{R\s{F}}\!\neq\!\rank{E\s{F}}$, 
there will exist convex combinations of rows of $R\s{F}$ or convex combinations of columns of $E\s{F}$ 
that cannot be separated by fragments' preparations or measurements; see the proof of Lemma~2 of Ref.~\cite{shahandeh2025unified}.
This can only be removed via additional preparations (columns of $E$ not within $E\s{F}$) or additional measurement outcomes (rows of $R$ not within $R\s{F}$) of the parent theory, which contradict the relative tomographic completeness of $C\s{F}$.
This contradiction forces $R\s{F}$ and $E\s{F}$ to have equal ranks, and relative tomographic completeness then fixes the rank of $C\s{F}$ such that $\rank{C\s{F}}\!=\!\rank{R\s{F}}\!=\!\rank{E\s{F}}$.
Hence, $C\s{F}\!=\!R\s{F}E\s{F}$ is an ENMF of $C\s{F}$.
The latter contradicts the assumption that $C\s{F}$ does not admit an ENMF. \qed 

When considering specific information-processing tasks such as cloning, the obtainable statistics usually do not even imply relative tomographic completeness.
Therefore, one cannot directly infer the fragment’s contextuality from the rank separation of the task’s probability matrix.
Our second main result bridges this gap.

\paragraph*{Theorem~2.}
Suppose $C\s{F}$ is a fragment COPE matrix of another fragment COPE matrix $C'\s{F}$.
Suppose $\rank{C\s{F}}\!=\!\rank{C'\s{F}}\!=\!r$ and that $C\s{F}$ does not admit an ENMF.
Then, $C'\s{F}$ does not admit an ENMF.
\paragraph*{Proof.}
Let $C'\s{F}\!=\!R'\s{F}E'\s{F}$ be an ENMF of $C'\s{F}$ with inner dimension $\mathbb{k}$.
Then, if we choose a submatrix $C\s{F}$ of $C'\s{F}$ such that $\rank{C\s{F}}\!=\!\rank{C'\s{F}}\!=\!r$, using monotonicity of rank, it must be that $R'\s{F}$ and $E'\s{F}$ have submatrices $R\s{F}$ and $E\s{F}$, respectively, such that $\rank{R'\s{F}}\!=\!\rank{R\s{F}}\!=\!r$, $\rank{E'\s{F}}\!=\!\rank{E\s{F}}\!=\!r$, and $C\s{F}\!=\!R\s{F}E\s{F}$.
The latter is an ENMF of $C\s{F}$ in $\mathbb{k}$ dimensions, which is a contradiction. \qed

By choosing $C'\s{F}$ to be the fragment COPE matrix of the relatively tomographically complete fragment and concatenating Theorems~2 and~1, we can deduce the rank separation of the operational theory from the rank separation of the task's fragment COPE matrix.
Theorems~1 and~2 are thus central to our contextuality proof technique via rank separation:
Given a task's fragment COPE matrix $C\s{F}$, we establish a lower bound on the rank of $R\s{F}$ or $E\s{F}$ and compare it with the rank of the fragment COPE matrix of the relatively tomographically complete fragment, $C'\s{F}$. A gap between these bounds proves, via Theorems~1 and~2, that Eq.~\eqref{eq:nctxt_rank} cannot be satisfied.

Such lower bounds were derived in Ref.~\cite{shahandeh2025unified} for nonnegative matrices with a particular sparsity pattern, as follows.

\paragraph*{Theorem~3.~\cite{shahandeh2025unified}}
Suppose $C\s{F}$ is a (fragment) COPE matrix for $n$ preparations and $n$ measurement outcomes such that each row (column) of $C\s{F}$ has at least one zero entry that differs from all other rows (columns), denoting the impossibility of a specific measurement outcome given a particular preparation.
Then, the RF or ES matrix spans an $l$-dimensional subspace of $\V\s{ont}$ for the largest integer $l$ such that $\binom{l}{\lfloor l/2\rfloor}\!\leq\! n$.

We now present the rank‑separation proof technique used to establish contextuality in quantum information–processing tasks. Our first case-study is the minimum-error quantum state discrimination (MEQSD)~\cite{Schmid2018}. We use this as a warm-up example for our main result for cloning.
Here, the goal is to discriminate between two pure quantum states $\ketbra{\psi}$ and $\ketbra{\phi}$ in a single projective measurement $\{\ketbra{g_\phi},\ketbra{g_\psi}\}$.
Assuming that the probability of guessing each state correctly in this measurement is symmetric, the quantity $s\s{q}:=|\braket{\phi}{g_\phi}|^2\!=\!|\braket{\psi}{g_\psi}|^2$ represents the success probability.
The confusability of the two states, on the other hand, is given by $c\s{q}\!:=\!|\braket{\phi}{\psi}|^2$.
Schmid and Spekkens~\cite{Schmid2018} show that the quantum dependence of the success probability on the confusability cannot be reproduced by any noncontextual ontological model.
Here, we recover this result with a much shorter proof.

Consider the nontrivial cases in which $c\s{q}\!:=\!|\braket{\phi}{\psi}|^2 \!\neq\! 1$.
The states for the optimal MEQSD strategy are $\St\s{q}:=\{\ketbra{\phi},\ketbra{\phi^\perp},\ketbra{\psi},\ketbra{\psi^\perp},\ketbra{g_\phi},\ketbra{g_\psi}\}$.
They live in a 3-dimensional space as they belong to the equatorial plane of the Bloch sphere so that $\dim \spn \St\s{q} \!=\! 3$.
The corresponding effects are $\EE\s{q}\!=\!\St\s{q}$ which similarly span a 3-dimensional space.
These imply that the rank of the $6\!\times\!6$ fragment COPE matrix obtained in MEQSD,
\begin{gather}\label{app:eq:MOP_MEQSD}
\small{C\s{MEQSD}} = 
\begin{pmatrix}
        1 & 0 & c\s{q} & 1-c\s{q} & s\s{q} & 1-s\s{q} \\
        0 & 1 & 1-c\s{q} & c\s{q} & 1-s\s{q} & s\s{q} \\
        c\s{q} & 1-c\s{q} & 1 & 0 & 1-s\s{q} & s\s{q} \\
        1-c\s{q} & c\s{q} & 0 & 1 & s\s{q} & 1-s\s{q} \\
        s\s{q} & 1-s\s{q} & 1-s\s{q} & s\s{q} & 1 & 0 \\
        1-s\s{q} & s\s{q} & s\s{q} & 1-s\s{q} & 0 & 1
\end{pmatrix},
\end{gather}
is three.
However, it follows from Theorem~3 that the rank of one of RF or ES matrices is at least four, establishing the rank separation in the task.
By Theorem~2, the fragment defined by the equatorial plane of the Bloch sphere does not admit an ENMF. 
Consequently, Theorem~1 implies that the full operational qubit theory also fails to admit an ENMF, and is therefore contextual.
The contrapositive of this reasoning makes explicit how the operational theory's contextuality manifests as a resource for the MEQSD task.
Furthermore, if $c\s{q}=0$, i.e., the input states are orthogonal and perfectly distinguishable, then both the rank of $C\s{MEQSD}$ and the ontic dimensionality reduce to two, and the noncontextual model performs as well as quantum mechanics.

We can also prove that the optimal success probability in state discrimination as a function of confusability obtained by imposing contextuality on its COPE is equal to the optimal quantum success probability. This further demonstrates the effectiveness of rank separation. Interestingly, the theory-agnostic nature of our approach then suggests that the quantum success probability is optimal among \textit{all} GPTs. To this end, given that the size of the ontological models for $C\s{MEQSD}$ in Eq.~\eqref{app:eq:MOP_MEQSD} is bounded from below by four, we impose the condition $\rank{C\s{MEQSD}}\leq 3$.
Writing out the determinant of the matrix and solving it for $s\s{q}$ we find for the nontrivial case that $s\s{q}(c\s{q})\!=\!(1\pm\sqrt{1 \!-\! c\s{q}})/2$.
The plus solution optimizes the success probability and is exactly the quantum bound given by the Helstrom measurement~\cite{Helstrom1969}.
Notably, the contextual solutions also include the nonoptimal (and complementary) function $s\s{q2}\!=\!1-s\s{q1}$.
Interestingly, because we did not make any \textit{a priori} assumption about the desired contextual GPT, the optimal bounds above show that quantum theory provides the optimal state-discrimination performance among all of them.

We now show our central result on the contextuality of Quantum (deterministic) Approximate Cloning Machines (QACM), in which the goal is to find the optimal unitary cloner for a given family of quantum states~\cite{Bruss1998}.
It has been shown that the optimal quantum cloner of two states, known as the \textit{state-dependent} cloner, is contextual~\cite{lostaglio2020contextual}.
This leaves us with the intriguing open question of whether the optimal QACM is fundamentally a contextual operation, which underscores the potential significance of quantum contextuality as a resource for quantum cryptography~\cite{Bennet1992,Ambainis2016}, among many other applications.
We answer this open question affirmatively.

\paragraph*{Theorem~4.}
A noncontextual ontological model cannot achieve optimal global fidelity in phase-covariant and universal quantum cloning.
\paragraph*{Proof.}
The COPEs in both cases have ranks no more than 16.
Therefore, we choose $\lceil\frac{1}{4}\binom{16+2}{9}\rceil\!=\! 12155$ pure states (on the Bloch sphere's equator or anywhere on it for phase-covariant and universal cloning, respectively).
We chose the preparations to be the collection of ideal clones of these states, their optimal clones, and their orthogonal vectors.
The COPE for this choice of vectors is a $48620 \!\times\! 48620$ matrix with a rank of at most 16 and, using Theorem~1 an ontological dimension of at least $18$.
By Theorem~3, the RF or ES spans the 18-dimensional ontic vector space.
Hence, there is a rank separation between the COPE and either $E$ or $R$, implying via Theorems~2 and~1 that any ontological model for the optimal phase-covariant and universal quantum cloning is contextual.  
\qed

Using the rank-separation technique, we also recover a similar result for state-dependent cloning previously shown in Ref.~\cite{lostaglio2020contextual}; see END MATTER~\ref{app:qsdc}.
Here, we leverage this approach to explain the gap between the optimal global fidelities in noncontextual versus contextual cloners.
This gap highlights the role of contextuality as a resource in achieving optimal cloning fidelities. 

First, note that the average global fidelity of each cloner is the average of global fidelities $\{s_i \!:=\! |\braket{aa_i}{\alpha_a}|^2\}$ which represent the probabilities that each cloned state $\ketbra{\alpha_i}$ passes the test of the ideal clone $\ketbra{aa_i}$.
These overlaps, together with the pairwise global distinguishing error probability, or \textit{confusability} in the contexts of MEQSD and state-dependent cloning~\cite{Schmid2018,lostaglio2020contextual}, $c_{ij}\!=\!|\braket{a_i}{a_j}|^4$, characterizes the different input states. 

Now, assume that $C$ is the COPE for a noncontextual cloning protocol, i.e., the condition in Eq.~\eqref{eq:nctxt_rank} is met, and $C$, $R$, and $E$ are full rank.
In this case, each $s_i$ (therefore the global fidelity) is either a linear function of confusabilities $\{c_{ij}\}$ or a constant.
However, if the COPE was contextual, the rank separation would have dictated that, for the same NMF, $C$ has a reduced rank, meaning that its determinant is zero.
For matrices of size larger than two, this implies a strictly nonlinear dependency between fidelity and confusabilities, hence a gap between fidelities in noncontextual and contextual protocols.
This connection between contextuality and optimal fidelity is further exemplified by the optimal fidelity achieved in state-dependent QACM, where the clones necessarily lie within the span of the ideal clones~\cite{Bruss1998}.
We thus conclude that all cloning machines achieving fidelities beyond classical cloners, including suboptimal cloners, are contextual.

\paragraph*{Discussions.}
We demonstrated the application of the rank separation criterion~\cite{shahandeh2025unified} as a powerful method to verify generalized contextuality in information-processing protocols. We proved that any optimal quantum cloning machine is inherently contextual, resolving an open problem posed by Lostaglio and Senno~\cite{lostaglio2020contextual}, and answering a fundamental question about the relationship between cloning and nonclassicality. As another application, we showed that our method not only demonstrates contextuality in MEQSD~\cite{Schmid2018} but also establishes a converse relationship: imposing contextuality on the protocol yields an advantage.
The theory-agnostic nature of our approach also showed that the maximum advantage obtainable in state discrimination of two-level systems by contextual GPTs is achieved by quantum mechanics.
We showed a similar result for the state-dependent cloning protocol~\cite{lostaglio2020contextual}. Furthermore, our approach gives a deeper intuition into the role of contextuality in cloning tasks by revealing that the class of contextual cloning operations is larger than those resulting in maximal fidelity.

Looking ahead, a natural direction is to extend our methods to other information-processing scenarios in which the role of contextuality remains unclear. Cryptographic protocols are a particularly compelling candidate, as our methods may help identify the quantum resources underlying security guarantees, especially where quantum communication provides an advantage or enables functionalities with no classical counterpart. Beyond cryptography, our approach may also be applied to broader computational and learning tasks, including quantum learning, natural language processing, and communication complexity. At the same time, certain limitations should be noted. Current proofs of rank separation rely on zero-valued entries—outcomes that never occur—whereas in practice, noise and imperfections only make such outcomes highly improbable. This may be addressed numerically using approximate nonnegative matrix factorizations~\cite{Berry2007}, though analytical lower bounds in the approximate setting remain scarce~\cite{Dewez2022}. Finally, while our results remain valid for countably infinite collections of preparations and effects~\cite{shahandeh2025unified}, extending rank separation to continuous-variable systems remains an important open direction. We conjecture that the framework should still hold there, with a suitable reformulation in terms of integral operators and their decompositions. Such an extension would be highly relevant for continuous-variable quantum information processing, where quantum communication and metrology provide natural testbeds for studying the interplay between contextuality and nonclassical advantages.

\acknowledgements
FS and MD gratefully acknowledge the financial support from the Engineering and Physical Sciences Research Council (EPSRC) through the Hub in Quantum Computing and Simulation grant (EP/T001062/1). FS also acknowledges support and resources provided by the Royal Commission for the Exhibition of 1851 Research Fellowship. TY acknowledges the support through the Quantum Computing Studentship funded by Royal Holloway, University of London.

\bibliography{quantum}
\bibliographystyle{apsrev4-1}

\begin{widetext}


\setcounter{equation}{0}
\renewcommand{\theequation}{\thesection.\arabic{equation}}

\newpage
\section{End Matter}

\subsection{Quantum approximate state-dependant cloning}\label{app:qsdc}

In this section, we provide proof of the contextuality of optimal state-dependent cloning using the rank separation approach.

\begin{theorem*}\label{th:SDCL_Context}
    A noncontextual ontological model cannot achieve optimal global fidelity in state-dependent quantum cloning.
\end{theorem*}
\begin{proof}
The set of preparations relevant to state-dependent cloning is
$\St=\{\ketbra{aa},\ketbra{(aa)^\perp},\ketbra{bb},\ketbra{(bb)^\perp},\ketbra{\alpha},\\
\ketbra{\alpha^\perp},\ketbra{\beta},\ketbra{\beta^\perp}\}$, where $\ketbra{aa}$ and $\ketbra{bb}$ are the ideal clones, $\ketbra{\alpha}$ and $\ketbra{\beta}$ are the optimal bipartite clones.
It is known that the vectors $\ket{\alpha}$ and $\ket{\beta}$ live in the two-dimensional space spanned by $\ket{aa}$ and $\ket{bb}$~\cite{Bruss1998}.
Hence, without loss of generality, we choose the orthogonal vectors $\ket{(aa)^\perp}$, $\ket{(bb)^\perp}$, $\ket{\alpha^\perp}$, and $\ket{\beta^\perp}$ to lie within the same two-dimensional space.
The effects are also given by the same projectors, $\EE=\St$.
We can now write the $8 \times 8$ COPE $C\s{SDC}$ for the state-dependent cloning.
Given that all vectors are in $\spn\{\ket{aa},\ket{bb}\}$, the rank of $T\s{SDC}$ is at most three.
However, using Theorem~1.~(b), any ontological model generating the same statistics lives in an (at least) five-dimensional vector space because $8 < 10=\binom{5}{3}$, $8 < \binom{5}{\lfloor 5/2\rfloor}=\binom{5}{2}=10$, and 5 is the smallest integer where this inequality is satisfied.
Also, by Theorem~1, one of the spaces of epistemic states or response functions must span a $4$-dimensional ontic subspace because $8>6=\binom{4}{2}$.
Hence, Eq.~(3) cannot hold, and any ontological model for the optimal state-dependent cloning must be contextual.    
\end{proof}
The COPE used in the above theorem can be written out as follows (which is not the column-stochastic version):
    \begin{gather}\label{eq:COPE_SDCLONE}
    C\s{PCQC}=
    \begin{bmatrix}
        1 & 0 & p_f & 1-p_f & \delta & 1-\delta & p_i & 1 - p_i \\
        0 & 1 & 1-p_f & p_f & 1-\delta & \delta & 1 - p_i & p_i \\
        p_f & 1-p_f & 1 & 0 & p_i & 1-p_i & p_s & 1 - p_s \\
        1-p_f & p_f & 0 & 1 & 1-p_i & p_i & 1-p_s & p_s \\
        \delta & 1-\delta & p_i & 1-p_i & 1 & 0 & p_f & 1-p_f\\
        1-\delta & \delta & 1-p_i & p_i & 0 & 1 & 1-p_f & p_f \\
        p_i & 1 - p_i & p_s & 1 - p_s & p_f & 1 - p_f & 1 & 0 \\
        1-p_i & p_i & 1-p_s & p_s & 1-p_f & p_f & 0 & 1 
    \end{bmatrix},
    \end{gather}
where
\begin{equation}
    p_s = |\braket{\alpha}{\beta}|^2, \quad
    p_f = |\braket{\alpha}{aa}|^2, \quad
    p_i = |\braket{\alpha}{bb}|^2, \quad
    \delta = |\braket{aa}{bb}|^2.
\end{equation}
Using Theorem~1.~(b) of the main text, if we do not assume quantum mechanics \textit{a priori} and impose contextuality, it must hold that $\rank{C\s{PCQC}}\leq 3$, assuming the largest rank of the matrix factors satisfies the inequality in Theorem 1b.
Computing the eigenvalues of $C\s{PCQC}$ as functions of $p_s$, $p_f$, $p_i$, and $\delta$, it turns out that for normalized vectors $\ket{\alpha}$ and $\ket{\beta}$ that lie within the span of the ideal clones $\ket{aa}$ and $\ket{bb}$ we obtain the desired result.
This is in agreement with the optimal quantum case~\cite{Bruss1998} in which,
\begin{equation}
    \ket{\alpha} = \mu \ket{aa} + \eta \ket{bb} \quad\quad \ket{\beta} = \eta \ket{aa} + \mu \ket{bb}.
\end{equation}
Note that $\mu$ and $\eta$ ensure that the clones are symmetric, i.e., the fidelities of the clones are the same for both possible choices of states. 
If we change the variables $\mu = c + d$ and $\eta = c - d$, the normalization condition of the clone states will be given by
\begin{equation}
    (c+d)^2 + (c-d)^2 + 2(c^2 - d^2) \cos{\phi}^2 = 1,
\end{equation}
where $\braket{a}{b} = \cos{\phi}$.

Finding solutions that reduce the rank of the COPE is generally challenging due to the nonlinearity of the problem.
Nevertheless, let us fix $\braket{a}{b}=1/\sqrt{2}$ ($\phi = \pi/4$) and look for values of $c$ and $d$ that reduce the rank.
The set of solutions includes
\begin{equation}
    c = \sqrt{\frac{1}{6}(1 + \frac{1}{\sqrt{2}})} \quad \quad d = \sqrt{\frac{1}{6}(1 - \frac{1}{\sqrt{2}})}
\end{equation}
resulting in the set of eigenvalues $\{4, 5/2, 3/2, 0, 0, 0, 0, 0\}$. 
These values precisely correspond to the form of cloned states that saturate the maximum fidelity for state-dependent cloning of states with $1/\sqrt{2}$ overlap~\cite{Bruss1998}.
It thus follows that the optimal quantum clones are necessarily contextual.
However, there exist a large family of solutions that guarantee contextuality and yet do not correspond to optimal cloning strategies.

\end{widetext}

\end{document}